\documentclass[11pt,letterpaper]{article}
\usepackage{amsmath,amssymb,amsthm}
\usepackage{fullpage}
\usepackage{xspace}
\usepackage{multirow}
\usepackage{ifpdf}

\usepackage{fancyhdr}
\usepackage[hypertex]{hyperref}

\newcommand{\REMOVE}[1]{}

\newtheorem{theorem}{Theorem}[section]
\newtheorem{lemma}[theorem]{Lemma}

\newtheorem{corollary}[theorem]{Corollary}

\theoremstyle{definition}
\newtheorem{definition}[theorem]{Definition}

\newcommand{\oneplus}{\oplus_1}
\newcommand{\twoplus}{\oplus_2}
\newcommand{\threeplus}{\oplus_3}
\newcommand{\symd}{\bigtriangleup}

\newcounter{this-list}

\newcounter{par-list}
\newlength{\parlistlength}

\begin{document}

\begin{titlepage}

\title{Matroid Secretary for Regular and Decomposable Matroids}
\author{Michael Dinitz\\Weizmann Institute of Science\\{\tt mdinitz@cs.cmu.edu} \and Guy Kortsarz\\Rutgers University, Camden\\{\tt guyk@camden.rutgers.edu}}
\date{}

\maketitle

\begin{abstract}
In the \emph{matroid secretary problem} we are given a stream of elements and asked to choose a set of elements that maximizes the total value of the set, subject to being an independent set of a matroid given in advance.  The difficulty comes from the assumption that decisions are \emph{irrevocable}: if we choose to accept an element when it is presented by the stream then we can never get rid of it, and if we choose not to accept it then we cannot later add it.  Babaioff, Immorlica, and Kleinberg~[SODA 2007] introduced this problem,  gave $O(1)$-competitive algorithms for certain classes of matroids, and conjectured that every matroid admits an $O(1)$-competitive algorithm.  However, most matroids that are known to admit an $O(1)$-competitive algorithm can be easily represented using graphs (e.g.~graphic, cographic, and transversal matroids).  In particular, there is very little known about \emph{$F$-representable} matroids (the class of matroids that can be represented as elements of a vector space over a field $F$), which are one of the foundational types of matroids.  Moreover, most of the known techniques are as dependent on graph theory as they are on matroid theory.  We go beyond graphs by giving $O(1)$-competitive algorithms for \emph{regular} matroids (the class of matroids that are representable over any field), and use techniques that are fundamentally matroid-theoretic rather than graph-theoretic.

Our main technique is to leverage the seminal regular matroid decomposition theorem of Seymour, which gives a method for decomposing any regular matroid into matroids which are either graphic, cographic, or isomorphic to a simple $10$-element matroid.  We show how to combine in a black-box manner \emph{any} algorithms for these basic classes into an algorithm for a given regular matroid, i.e.~how to respect the decomposition.  In fact, this allows us to generalize beyond regular matroids to \emph{any} class of matroids that admits such a decomposition into classes for which we already have good algorithms.  In particular, we give an $O(1)$-competitive algorithm for the class of \emph{max-flow min-cut} matroids, which Seymour showed can be decomposed into regular matroids and copies of the Fano matroid.
 \end{abstract}

\thispagestyle{empty}

\end{titlepage}

\section{Introduction}

In the classical secretary problem we are asked to perform the seemingly-simple task of selecting the most valuable element of a set, but in the setting where elements are presented online (with their values) and where our accept/reject decision for each element is irrevocable.  That is, we see the elements and their values one at a time, and upon seeing an element must decide whether to choose it, with this decision being irrevocable.  It is easy to see that if the values and the online order are both adversarial, then no algorithm can do very well.  However, if the ordering is random then it is well-known that the optimal solution (under adversarial values) is to let the first $1/e$ elements go by, and then select the next element that is more valuable than the best of the first $1/e$ elements.  This selects the most valuable element with probability at least $1/e$~\cite{Dynkin}.

One way of extending this problem which has received considerable attention is to allow ourselves to choose a set of elements, rather than a single one.  Kleinberg~\cite{Kle05} considered perhaps the simplest version of this problem, where we are given an integer $k$ and want to pick the set with at most $k$ elements with as much total value as possible.  He showed how to generalize the classical algorithm so that as $k$ gets larger, the competitive ratio improves.  This was then further generalized by Babaioff, Immorlica, and Kleinberg~\cite{BIK07}, who defined the \emph{matroid secretary problem} to be the version in which the sets that we are allowed to select are the independent sets of a matroid (and as before we are attempting to choose the set of maximum value).  They showed a connection between algorithms for the matroid secretary problem and truthful mechanisms for online auctions over the corresponding matroid domains, gave an $O(\log \rho)$-competitive algorithm for matroids of rank $\rho$, and for several specific classes of matroids (uniform, graphic, and bounded-degree transversal) gave $O(1)$-competitive algorithms.  Moreover, they conjectured that there is an $O(1)$-competitive algorithm for the matroid secretary problem on the class of all matroids.

Since~\cite{BIK07} there has been a significant amount of work towards this conjecture.  Chakraborty and Lachish~\cite{CL12} gave an improved $O(\sqrt{\log \rho})$-competitive algorithm, and some other classes of matroids have been shown to admit $O(1)$-competitive algorithms (e.g.~transversal matroids~\cite{DP12,KP09}, cographic matroids~\cite{Soto11}, and laminar matroids~\cite{IW11}).  However, matroids based on linear algebra (i.e.~\emph{representable} or \emph{vector} matroids) have resisted all attacks: not only is there nothing known for the general vector matroid case (other than the $O(\sqrt{\log \rho})$ algorithm for general matroids), there is not even anything known for any subclasses of vector matroids other than those that can be directly represented as graphs (e.g.~graphic and cographic matroids, both of which are easily seen to be vector matroids).  In this paper we give the first $O(1)$-competitive algorithm for classes of vector matroids that do not have a simple graph interpretation, namely \emph{regular} matroids and some extensions.  Regular matroids have multiple equivalent definitions, but are inherently linear algebraic: one definition is that they are the matroids that can be represented by totally unimodular matrices over $\mathbb{R}$, and another definition is that they are the matroids that are representable over any field.  Thus we take the first step towards a better understanding of secretary algorithms for vector matroids.

In fact, not only are all specific classes of matroids that are known to have $O(1)$-competitive algorithms based on graphs, the techniques used are generally based more on graph theory than on matroid theory.  We give the first algorithms that are based on deep results in matroid theory (Seymour's decomposition theorem for regular matroids and applications of splitter theory) rather than graph theory.  Our main technique is to use the theory of \emph{matroid decompositions} as pioneered by Seymour~\cite{Sey80} and Truemper~\cite{Tru92}, which give algorithms for decomposing regular matroids (and larger classes) into the sums of simpler matroids.  It turns out that these simpler matroids are all classes for which we already have good algorithms (graphic, cographic, or constant size matroids), and thus we simply need to design an algorithm that respects the notion of sum used in the decomposition.  While this is intuitively simple, it turns out to be quite nontrivial: we would like to be able to operate independently in each of the simple matroids, but it turns out that in order to do this we need global coordination between them to maintain independence.  We achieve this global coordination by modifying the decomposition to have a particularly nice structure, and then using this structure to make global decisions about how to modify each of the simple matroids in order to ensure that running a matroid secretary algorithm in each of them independently guarantees us a set that is independent in the sum.  We also need to guarantee that the modification of each matroid is relatively small, in order to make sure that the competitive ratio of our algorithm is not much worse than the worst of the algorithms that we run on the simpler matroids.

\subsection{Our Results}

In Section~\ref{sec:sums} we formally define the notion of matroid sums that we use, which was originally defined by Seymour~\cite{Sey80}.  Our main theorem is that if we are given a matroid $M$ and a decomposition of it into sums of matroids that we already have good algorithms for, then we have a good algorithm for $M$.  There is a slight technical complication: we need good algorithms not just for the matroids that are summed, but good algorithms even when some elements are contracted (see Section~\ref{sec:basic_matroid} for a definition of matroid contraction) and elements are added in parallel (two elements are parallel in a matroid if they form a circuit).

\begin{theorem} \label{thm:intro_main}
Suppose that we are given a matroid $M$ and a decomposition of $M$ into the sum of matroids from some set $\mathcal M$.  Let $\mathcal M'$ denote the class of matroids that can be obtained from $\mathcal M$ by taking a matroid $M' \in \mathcal M$, contracting either a single element or a $3$-circuit, adding parallel elements, and deleting elements.  Further suppose that we have an $\alpha$-competitive algorithm for all matroids in $\mathcal M' \cup \mathcal M$.  Then there is a $3\alpha$-competitive algorithm for the matroid secretary problem on $M$.
\end{theorem}

Seymour's famous decomposition theorem for regular matroids~\cite{Sey80} implies that regular matroids can always be decomposed (in polynomial time) into graphic matroids, cographic matroids, and matroids isomorphic to a special $10$-element matroid $R_{10}$.  Graphic matroids are known to admit $2e$-competitive algorithms~\cite{KP09}, cographic matroids are known to admit $3e$-competitive algorithms~\cite{Soto11}, and we give a simple $20e/9$-competitive algorithm for $R_{10}$ (and matroids that can be constructed from it as described in Theorem~\ref{thm:intro_main}).  Thus we get the following corollary, which was the main motivation for this work:

\begin{corollary}
There is a $9e$-competitive algorithm for the matroid secretary problem on regular matroids.
\end{corollary}

While regular matroids were the original motivation for this approach, because our theorem is so general it turns out that we get even broader classes ``for free."  The most important one is the class of \emph{max-flow min-cut} (MFMC) matroids, which are defined in Section~\ref{sec:classes}.  These are the matroids for which a natural generalization of the max-flow min-cut theorem for graphs continues to hold, and it turns out that they are a generalization of regular matroids.  In fact, as a corollary of splitter theory~\cite{Sey80,Sey95} and an earlier characterization~\cite{Sey77}, Seymour proved that any MFMC matroid can be constructed by taking sums of regular matroids and the Fano matroid $F_7$.  We give a simple $(7e/3)$-competitive algorithm for matroid secretary on $F_7$ (and matroids that can be constructed from it by contraction, deletion, and addition of parallel elements), giving the following corollary.

\begin{corollary}
There is a $9e$-competitive algorithm for the matroid secretary problem on MFMC matroids.
\end{corollary}

\paragraph{Our techniques.}
Seymour's decomposition theorem for regular matroids gives us a way of decomposing any regular matroid into into matroids that are graphic, cographic, or isomorphic to $R_{10}$.  Given a decomposition, these underlying matroids are called the \emph{basic} matroids of the decomposition.  Unfortunately this decomposition has a somewhat complex structure, and simply running a good algorithm on each matroid in the decomposition in parallel does not work (it either gives a set that is not independent, or if constrained to only return an independent set does not return sets with large enough value).  It is not hard to construct an algorithm that always returns independent sets but loses a constant at every level of the decomposition, but this is not sufficient as the decomposition might have depth linear in the number of elements.  What we'd like is an algorithm that only loses a constant \emph{globally}, i.e.~is only a constant worse than the worst of the algorithms for each of the basic matroids.

In order to do this we need to utilize the global structure of the decomposition, and use this structure to make local decisions that somehow imply global competitiveness.  We do not know how to do this for arbitrary decompositions, but we show how to modify any given decomposition into one with global structure that will allow us to do this.  We first modify the decomposition to have the property that whenever two matroids are summed in the decomposition, the elements in the intersection of the matroids are all completely contained in two basic matroids.  We accomplish this by showing how to shift elements around in a decomposition until this is true, without affecting the validity of the decomposition.   Once we have such a decomposition, we prove that it has a very nice global structure: if we put an edge between two basic matroids that have intersecting ground sets, the resulting graph is a forest.  We then root each tree in this forest and use the resulting partial order to define a modification for each basic matroid that will allow us to simply run (in parallel) a separate algorithm in each basic matroid while guaranteeing global independence and losing at most a constant factor overall.

\section{Preliminaries}

\subsection{Basic Matroid Theory} \label{sec:basic_matroid}
We begin with some basic definitions and results from matroid theory.  We follow the conventions used by Oxley is his excellent book on matroid theory~\cite{Oxley}.  For a matroid $M$, we will let $E(M)$ denote the ground set of $M$.  A collections of sets $\mathcal I \subseteq 2^{E(M)}$ form the independent sets of a matroid if they satisfy the following three axioms: $\emptyset \in \mathcal I$, if $I \in \mathcal I$ and $I' \subseteq I$ then $I' \in \mathcal I$, and if $I_1, I_2 \in \mathcal I$ with $|I_1| < |I_2|$ then there is some element $e \in I_2 \setminus I_1$ such that $I_1 \cup \{e\} \in \mathcal I$.  A set that is not independent is called \emph{dependent}.  A \emph{base} or \emph{basis} is a maximal independent set (all bases clearly have the same cardinality), and a \emph{circuit} is a minimal dependent set.  We will call a circuit of size $k$ a \emph{$k$-circuit}, a circuit of size $1$ a \emph{loop}, and if two elements for a $2$-circuit then we say that they are \emph{parallel}.  We will frequently define matroids not in terms of their independent sets but in terms of their circuits.

Given an $m \times n$ matrix $A$ with entries over a field $F$, the collection of sets of columns of $A$ that are linearly independent in the $m$-dimensional vector space over $F$ form the independent sets of a matroid known as the \emph{vector matroid} of $A$.  That is, a set of columns is independent if and only if its elements are linearly independent.  A matroid is said to be representable over $F$, or $F$-representable, if it is isomorphic to the vector matroid of a matrix over $F$.  A \emph{binary} matroid is a matroid that is $GF(2)$-representable.  There are a few simple properties of binary matroids that will turn out to be very useful.  Given two sets $X$ and $Y$, the \emph{symmetric difference} $X \symd Y$ is defined to be $(X \cup Y) \setminus (X \cap Y) = (X \setminus Y) \cup (Y \setminus X)$.

\begin{lemma}[{\cite[Theorem 9.1.2]{Oxley}, \cite{Fou81}}] \label{lem:binary_basics}
The following statements are equivalent for a matroid $M$:
\begin{enumerate}
\item $M$ is binary.
\item If $C_1$ and $C_2$ are distinct circuits, then $C_1 \symd C_2$ is a disjoint union of circuits.
\item The symmetric difference of any set of circuits is a disjoint union of circuits.
\item For any two distinct circuits $C_1, C_2$ and elements $x,y \in C_1 \cap C_2$, there is a circuit contained in $(C_1 \cup C_2) \setminus \{x,y\}$ (the \emph{double-elimination} property).
\end{enumerate}
\end{lemma}

A \emph{cycle} of a binary matroid (not to be confused with a cycle of a graph) is a disjoint union of circuits, i.e.~a set which can be partitioned into circuits.  It is easy to see that for any circuit contained in a cycle, we can assume that this partition contains the circuit:

\begin{lemma} \label{lem:partition}
Let $X$ be a cycle of a binary matroid $M$, and let $Z \subseteq X$ be a circuit of $M$.  Then there is a partition of $X$ into circuits of $M$ so that $Z$ is a circuit in the partition.
\end{lemma}
\begin{proof}
Since $X$ is a cycle it is a disjoint union of circuits, so Lemma~\ref{lem:binary_basics}(3) implies that $X \symd Z = X \setminus Z$ is a disjoint union of circuits, i.e.~a cycle.  Thus $X \setminus Z$ can be partitioned into circuits, so this partition together with $Z$ forms a partition of $X$ into circuits.
\end{proof}

Two matroid operations that we will use heavily are \emph{deletion} and \emph{contraction}.  Given a matroid $M$ with independent sets $\mathcal I$, and subset $X \subseteq E(M)$, the \emph{deletion} of $X$ from $M$, denoted by $M - X$, is the matroid on $E(M) \setminus X$ in which a subset of $E(M) \setminus X$ is independent if and only if it is independent in $M$.  With deletion in hand, we can define \emph{restriction}: the restriction of $M$ to $X$ (denoted by $M|X$) is just $M - (E(M) \setminus X)$.  An operation that is dual to deletion is \emph{contraction}: the contraction of $X$ from $M$, denoted by $M / X$, is the matroid on ground set $E(M) \setminus X$ in which a set $I \subseteq E(M) \setminus X$ is independent if and only if there is a basis $B$ of $M | X$ such that $I \cup B$ is independent in $M$.  In fact, the choice of basis does not matter: let $B$ be an arbitrary basis for $M | X$.  Then a set $I$ is independent in $M / X$ if and only if $I \cup B$ is independent in $M$ (see~\cite[Proposition 3.1.8]{Oxley}).

\subsection{Classes of Binary Matroids} \label{sec:classes}
We will be particularly interested in a few subclasses of binary matroids.  Given a graph $G = (V,E)$, the \emph{cycle matroid} of $G$ is the matroid with ground set $E$ in which a subset $E'$ is independent if and only if the subgraph $(V, E')$ is acyclic.  A matroid $M$ is said to be \emph{graphic} if there is a graph $G$ such that $M$ is isomorphic to the cycle matroid of $G$.  It is easy to see that every graphic matroid is also binary.

A different matroid defined on a graph $G = (V,E)$ is the \emph{bond matroid}, in which the ground set is $E$ (as with the cycle matroid) and a set $E' \subseteq E$ is independent if and only if $G \setminus E'$ has the same number of connected components as $G$.  In other words, $E'$ is independent if it does not contain an edge cut.  A matroid $M$ is said to be \emph{cographic} if there is a graph $G$ such that $M$ is isomorphic to the bond matroid of $G$.  It is not hard to see that the bond matroid of a graph $G$ is the dual of the cycle matroid (hence the term cographic) under the usual definition of matroid duality (see~\cite[Section 2.3]{Oxley} for details).  As with graphic matroids, it is easy to see that cographic matroids are binary.

A \emph{regular} matroid is a binary matroid that is representable over every field (an equivalent definition of a regular matroid is a matroid represented over $\mathbb R$ by a totally unimodular matrix~\cite[Theorem 6.6.3]{Oxley}).  It is easy to see that graphic and cographic matroids are also regular.

A non-obvious generalization of regular matroids are the \emph{max-flow min-cut} (MFMC) matroids.  Given a matroid $M$ and an element $e \in E(M)$, let $A$ be the matrix in which the columns are indexed by the elements of $E(M) \setminus \{e\}$ and the rows are indexed by the circuits of $M$ that contain $e$.  Let $\mathcal C^*(M)$ denote the cocircuits of $M$ that include $e$ (a set is a cocircuit if it a circuit in the dual matroid).  We say that $M$ is \emph{$e$-MFMC} if for any nonnegative integral weight vector $w$ defined on $E(M) \setminus \{e\}$ we have $\min_{C \in \mathcal C^*(M)} w(C\setminus \{e\}) = \max \{{\bf 1}^T y : y^T A \leq w^T, y \geq 0 \text{ and integral}\}$, where we slightly abuse notation and let $w(C \setminus \{e\}) = \sum_{u \in C \setminus \{e\}} w_u$.    This definition was introduced by Seymour~\cite{Sey77} as a generalization of the max-flow min-cut theorem for graphs: note that in a graph, if we add an edge $e$ between nodes $u$ and $v$ and then use $e$ as the special element, $\min_{C \in \mathcal C^*(M)} w(C\setminus \{e\})$ is the minimum $u,v$ cut and $\max \{{\bf 1}^T y : y^T A \leq w^T, y \geq 0 \text{ and integral}\}$ is the maximum flow under (integral) capacities $w$.  Seymour proved that a matroid is $e$-MFMC if and only if it is binary and has no $F^*_7$ minor containing $e$ (where $F^*_7$ is the dual of the Fano matroid $F_7$)~\cite{Sey77}.  A matroid is \emph{MFMC} if it is $e$-MFMC for all $e \in E(M)$, and so Seymour's theorem implies that a matroid is MFMC if and only if it is binary and contains no $F_7^*$ minor.

\subsection{Matroid Sums} \label{sec:sums}
We now define the notion of sums of binary matroids that we will be using.  This definition is originally due to Seymour~\cite{Sey80}.  If $M_1$ and $M_2$ are binary matroids, we define a new binary matroid $M = M_1 \symd M_2$ on ground set $E(M_1) \symd E(M_2)$.  The cycles of $M$ are all subsets of the form $C_1 \symd C_2$, where $C_1$ is a cycle of $M_1$ and $C_2$ is a cycle of $M_2$.  It is not hard to see that this does indeed define a binary matroid~\cite{Sey80}, in which the circuits are the minimal non-empty cycles and the independent sets are (as always) the sets that do not contain any circuit.  As in~\cite{Sey80}, we are only concerned with three special cases of this.
\begin{enumerate}
\item When $E(M_1) \cap E(M_2) = \emptyset$ and $E(M_1), E(M_2) \neq \emptyset$ we call $M$ the \emph{1-sum} of $M_1$ and $M_2$ and write $M = M_1 \oneplus M_2$.
\item When $|E(M_1) \cap E(M_2)| = 1$, let $\{z\} = E(M_1) \cap E(M_2)$.  If $z$ is not a loop or coloop of $M_1$ or $M_2$ and $|E(M_1)|, |E(M_2)| \geq 3$ we call $M$ the \emph{2-sum} of $M_1$ and $M_2$ and write $M = M_1 \twoplus M_2$.
\item When $|E(M_1) \cap E(M_2)| = 3$, let $Z = E(M_1) \cap E(M_2)$.  If $Z$ is a circuit of $M_1$ and $M_2$, and $Z$ does not include a cocircuit of either $M_1$ or $M_2$, and $|E(M_1)|, |E(M_2)| \geq 7$, then we call $M$ the \emph{3-sum} of $M_1$ and $M_2$ and write $M = M_1 \threeplus M_2$.
\end{enumerate}

If a binary matroid $M$ equals $M_1 \oplus_k M_2$ for some $k \in \{1,2,3\}$, then we write $M = M_1 \oplus M_2$.  While the above definition may be slightly confusing at first, it is useful to gain intuition by thinking about the \emph{clique-sum} of graphs.  Suppose we have two graphs $G_1$ and $G_2$, each of which contains a clique with exactly $k$ vertices.  Then we can identify the vertices of these two cliques, and remove some subset of the edges inside the (now single) clique.  This forms a graph $G$ which is a \emph{$k$-clique-sum} of $G_1$ and $G_2$.  These sums play an important role in structural graph theory and especially graph minor theory, and many theorems can be proved about them.  For example, any graph of treewidth at most $k$ can be decomposed into the clique-sum of graphs with at most $k+1$ vertices~\cite{Lov06}.

It is easy to see that if we slightly change the definition of clique-sum to require removing \emph{all} of the edges in the clique, the cycle matroid of the $1$-clique-sum of two graphs $G_1$ and $G_2$ is exactly the $1$-sum of the cycle matroid of $G_1$ and the cycle matroid of $G_2$, i.e.~the matroid notion of $1$-sum generalizes the graph notion of $1$-clique-sum.  Similarly, the matroid notions of $2$- and $3$-sums generalize the graph notions of $2$- and $3$-clique-sums.  This allows us to get intuition from the graphical case: the sum of two matroids should behave ``like" the clique-sum of two graphs.

With this intuition, it is not difficult to see that the extra restrictions in the definition (for $2$-sums that $|E(M_1)|, |E(M_2)| \geq 3$ and $z$ is not a loop or coloop, and for $3$-sums that $|E(M_1)|, |E(M_2)| \geq 7$ and $Z$ does not include a cocircuit) are there simply to ensure non-triviality.  The size constraints force the sum to be larger than the two original matroids, and the loop/coloop/cocircuit constraints simply ensure that that $M$ cannot also be expressed as a smaller sum (possibly modulo some deletions and contractions).  For example, if $M = M_1 \threeplus M_2$ but $Z = E(M_1) \cap E(M_2)$ contains a cocircuit $C^*$, then it is not hard to see that the fact that $M_1$ and $M_2$ are binary implies that $|C^*| = 2$ and $M = (M_1 \setminus C^*) \twoplus (M_2 \setminus C^*)$.  These non-triviality properties are important when proving structural theorems about decompositions, but for our purposes they are not important.  One of the central pieces of our algorithm is a method for modifying decompositions into sums into different decompositions that are ``nicer", and while it is easy to see that after such a modification the extra conditions can be restored by further modifications, in order to simplify the exposition we will essentially ignore the extra conditions.  So we will abuse notation and say that $M = M_1 \oplus M_2$ even if the extra conditions are not satisfied, with the understanding that we can always get back the extra conditions later if they turn out to be necessary (which for us they will not).

Informally, the notion of matroid decomposition that we use is decomposition into sums, i.e.~a matroid can be decomposed into some set of matroids if we can take sums of those matroids to get back the original matroid.  We define a decomposition into sums formally as follows.

\begin{definition} \label{def:decomposition}
A \emph{\{1,2,3\}-decomposition} of a matroid $\tilde M$ is a set of matroids $\mathcal M$ (called the \emph{basic matroids}) and a rooted binary tree $T$ in which $\tilde M$ is the root and the leaves are the elements of $\mathcal M$, with the property that every internal vertex is either the $1$-, $2$-, or $3$-sum of its children.
\end{definition}

Given a $\{1,2,3\}$-decomposition $T$, there is another graph that will prove to be extremely useful:

\begin{definition}
Given a $\{1,2,3\}$-decomposition $(T, \mathcal M)$ of a matroid $\tilde M$, the \emph{conflict graph} $G_T$ of the decomposition has vertex set $\mathcal M$ and an edge between $M_1, M_2 \in \mathcal M$ if $E(M_1) \cap E(M_2) \neq \emptyset$.
\end{definition}

The main motivation for considering these sums is Seymour's decomposition theorem for regular matroids.  The special matroid $R_{10}$ can be represented over $GF(2)$ as the $10$ vectors in the five-dimensional vector space over $GF(2)$ that have exactly three nonzero entries.

\begin{theorem}[\cite{Sey80}] \label{thm:regular_decomposition}
Every regular matroid $M$ has a $\{1,2,3\}$-decomposition in which every basic matroid is either graphic, cographic, or is isomorphic to $R_{10}$.  Moreover, such a decomposition can be found in time polynomial in the number of elements of $M$.
\end{theorem}

We note that this decomposition theorem is already known to have algorithmic implications: the fastest known algorithm for testing total unimodularity of a matrix runs in cubic time and is based on the ability to find such a decomposition if one exists~\cite{Tru90}.

While regular matroids are our motivation, we note that since our techniques are based on decompositions we can go a bit beyond regular matroids.  For example, MFMC matroids  can be decomposed (in polynomial time) into the $1$- and $2$-sums of regular matroids and copies of $F_7$~\cite[(7.6)]{Sey80}.  Since regular matroids can be further decomposed, this means that every MFMC matroid has a $\{1,2,3\}$-decomposition in which every basic matroid is either graphic, cographic, or isomorphic to $R_{10}$ or $F_7$.  Similarly, splitter theory~\cite{Sey80,Sey95} implies that the class of binary matroids that exclude an $F_7$ minor admit {1,2,3}-decompositions into regular matroids and copies of $F^*_7$ (the dual of $F_7$).

This type of decomposition and matroid sums more generally have been thoroughly studied by Truemper~\cite{Tru92}.  One particularly interesting aspect of this decomposition theorem is that there are already known constant-competitive algorithms for matroid secretary on graphic matroids (e.g.~\cite{BIK07,BDGIT09,KP09}) and on cographic matroids~\cite{Soto11}, and it is obviously trivial to be constant-competitive on matroids of constant size.  This means that if we can design a constant-competitive algorithm for 1-, 2-, and 3-sums of matroids for which we already have constant-competitive algorithms, then we have a constant-competitive algorithm for regular matroids.

\subsection{Algorithm Overview}

Before diving into the details, we give a brief overview of our algorithm and the intuition behind it.  Assuming that we have a good secretary algorithm for each basic matroid of a sum, how do we develop an algorithm that works for the sum?  Obviously we must use these algorithms in some way, but since we do not assume anything about how these algorithms work, we must use them in a black-box fashion.  We are forced to simply run the algorithms, but possibly on slightly modified matroids $M'_1$ and $M'_2$ (which still fall into classes for which we have good algorithms).  Note that if $M$ is the $2$- or $3$-sum of $M_1$ and $M_2$ then there is at least one element of each of $M_1, M_2$ that is not an element of $M$, so at a minimum we have to decide how to treat these ``fake" elements.

The most obvious thing to do is simply run the algorithm for each matroid on the set of elements that appear in the final matroid, i.e.~simply delete the fake elements and run our algorithm for $M_1$ on $M'_1 = M_1 |(E(M_1) \cap E(M))$ and our algorithm for $M_2$ on $M'_2 = M_2 |(E(M_2) \cap E(M))$, taking as our final solution the union of the two solutions.  Note that the element sets for $M'_1$ and $M'_2$ are disjoint, so we can just run the two algorithms in parallel in the obvious way (we run both algorithm simultaneously, and when seeing a new element we determine which matroid it belongs to and apply the algorithm for that matroid).  However, it is easy to see that the set we select might not be independent in $M$: suppose that $M = M_1 \twoplus M_2$ with $E(M_1) \cap E(M_2) = \{z\}$.  Then we might select a set $A_1 \subseteq E(M_1) \setminus \{z\}$ that is independent in $M_1$ but has the property that $A_1 \cup \{z\}$ is a circuit of $M_1$, and similarly we might select a set $A_2 \subseteq E(M_2) \setminus \{z\}$ that is independent in $M_2$ but $A_2 \cup \{z\}$ is a circuit of $M_2$.  Then by definition $A_2 \cup A_2$ is a cycle of $M$ and thus dependent (in fact, it is not hard to see based on the graphical intuition that $A_1 \cup A_2$ will in fact be a circuit of $M$).

The problem with this obvious approach was that by simply ignoring the fake element $z$, both $M_1$ and $M_2$ were able to select a set that, while it may have been independent, formed a circuit when $z$ was added.  This meant that the union of the two sets formed a circuit in the sum.  The obvious way to get around this would be for us to ``pretend" when running the two algorithms that $z$ is included, i.e.~force ourselves to pick a set from $M_1$ and a set from $M_2$ that are independent even with $z$ is included.  Technically, this corresponds to contracting $z$ in each matroid (i.e.~letting $M'_1 = M_1/\{z\}$ and letting $M'_2 = M_2 / \{z\}$).  More generally, instead of \emph{deleting} all fake elements, we could contract them.  It can be shown that this will always result in sets whose union is independent in $M$ (the intuition is clear from the graphical case), but unfortunately this is too pessimistic: this approach might not be able to recover enough value compared to $OPT$.  As an easy example of this, suppose that $M_1$ has an element $a$ that is parallel to $z$, i.e.~$\{a,z\}$ is a circuit in $M_1$.  Then we force ourselves to not include $a$ in our solution, but $\{a\}$ is clearly an independent set in $M$ and it could be the case that all of the weight is on the single element $a$.  Obviously this example can be worked around (since circuits of size $2$ are more of a technicality than a real problem), but more generally $M_1$ might have many fake elements (if there are many sums that involve elements of $M_1$), and contracting all of them might leave us with a matroid in which only the empty set is independent, even though $OPT$ is able to pick a large independent set from $M_1$.

To remedy this, note that we were being unnecessarily pessimistic: contracting $z$ in both $M_1$ and $M_2$ was sufficient to ensure that our eventual solution was independent in $M$, but it was not necessary.  In fact, it is not hard to see that if $z$ is contracted in \emph{either $M_1$ or $M_2$} the solution we return will be independent in $M$.  So we simply need to decide which matroid should contract $z$.  In the simple case of the sum of two matroids, this is easy -- we can just flip a coin and decide randomly.  But what should we do in larger sums?  It is not difficult to construct an instance in which randomly choosing is bad: for example, if $M_1$ has many fake elements, i.e.~is part of many sums.  Then in expectation about half of these fake elements will be contracted in $M_1$, which might still be an extremely large number.  They might even form a basis of $M_1$, which would imply that contracting them would still give the matroid in which only $\emptyset$ is independent.  In this case an adversary could put all of the value on elements in $E(M_1) \cap E(M)$, and we would not get a decent approximation.

So what we actually need is a way of assigning each element that is not in $E(M)$ (i.e.~the fake elements) to one of the (two) matroids in the sum that contain it.  We then contract it in the matroid to which it is assigned, and delete it from the other matroid containing it.  And we need to do this in a way that does not lose much of the total value.  As a way to build intuition, suppose that our decomposition is made up only of $2$-sums.  It turns out that in this case the conflict graph will always be a tree in which there is a bijection between edges and fake elements.  This means that we can root the tree arbitrarily, and for every fake element the two matroids that contain it are adjacent in this graph so one is the parent of the other.  Our algorithm is basically to have the child matroid contract the element and the parent matroid delete it.  For every sum in the decomposition the element summed along is contracted in one of the two matroids, so the eventual set we return is independent.  But since each matroid only contracts a single element, the value we get from that matroid is basically optimal (modulo some technical details involving circuits of size $2$).

While this algorithm works fine for $1$- and $2$-sums ($1$-sums change the conflict graph from a tree to a forest, and we can just root each tree in the forest arbitrarily), we run into problems when extending it to $3$-sums.  In particular, the conflict graph for a decomposition involving $3$-sums might not be a forest.  Fundamentally, this is because at some intermediate stage in the decomposition we might have summed along a circuit of size $3$ whose elements are not all contained in the same basic matroids.  With $2$-sums this was not an issue, since the sums are along a single element.  But with $3$-sums it is easy to construct examples where circuits of size $3$ are created at an intermediate stage that do not exist in any of the basic matroids.  Thus we need an extra step that, when given a decomposition involving $3$-sums for which the conflict graph is not a forest, modifies the decomposition so that the conflict graph becomes a forest while still being a valid decomposition.  We show how to do this by ``moving" fake elements between matroids in a consistent manner.  Then once we have a forest as a conflict graph, we simply root the trees and as before use the forest to coordinate which matroids should contract which fake elements.  This leads to a situation where every matroid only contracts at most $3$ elements, and it is easy to show that these contractions will not cost us too much.

\subsection{Lemmas about Sums}
We first state a few facts and prove a few simple lemmas about matroid sums.  We assume that many of these lemmas are already known, as they are generally quite simple and are not the main technical contribution of this paper, but we have not found a specific reference for them.

%
%
The first lemma basically states that if a set is contained in one of the original two matroids and does not touch the intersection of the two matroids that form the sum, then its independence is just determined by the matroid that contains it.

\begin{lemma} \label{lem:not_circuit}
Let $M = M_1 \oplus M_2$ with $E(M_1) \cap E(M_2) = A$, and let $Z \subseteq E(M_1) \setminus A$.  Then $Z$ is independent in $M$ if and only if it is independent in $M_1$.
\end{lemma}
\begin{proof}
We begin with the if direction, so let $Z$ be independent in $M_1$, and assume for contradiction that $Z$ is not independent in $M$, and let $W \subseteq Z$ be a circuit of $M$ contained in $Z$.  Since $Z$ is independent in $M_1$, so is $W$.  But since $W$ is a circuit of $M$, by definition it can be written as $X_1 \symd X_2$ where $X_1$ is a cycle of $M_1$ and $X_2$ is a cycle of $M_2$.  Note that $X_2 \subseteq A$ since $W \subseteq E(M_1)$, and thus $X_1$ consists of $W \cup X_2$.  And since $X_2$ must be a cycle of $M_2$ and $A$ is a circuit of $M_2$, either $X_2 = \emptyset $ or $X_2 = A$.  If $X_2 = \emptyset$ then we have a contradiction, since then $X_1 = W$ is independent in $M_1$ instead of being a cycle of $M_1$.  On the other hand, if $X_2 = A$ then we still have a contradiction.  In this case, since $X_1$ is a cycle containing $A$ by Lemma~\ref{lem:partition} it can be partitioned into circuits, one of which is $A$.  Thus $X_1 \setminus A = W$ is a cycle, so $W$ is dependent in $M_1$, contradicting its independence in $M_1$.  So $Z$ must be independent in $M$.

Now we prove the only if direction.  Let $Z$ be dependent in $M_1$.  Then it contains a circuit $W$ of $M_1$.  Let $X_1 = W$ and let $X_2 = \emptyset$.  Then $X_1$ is a cycle of $M_1$ and $X_2$ is a cycle of $M_2$, so $X_1 \symd X_2 = W$ is a cycle of $M$.  Thus $Z$ is dependent in $M$.
\end{proof}

The next lemma gives a condition on the creation of circuits of size $3$ that are in a sum but are split between the two component matroids.  Informally, it states that if a $3$-circuit is created by a sum, it must be formed of one element that formed a $2$-circuit with an element of the intersection, and two other elements that formed a $3$-circuit with the same element.

\begin{lemma} \label{lem:3circuit_creation}
Let $M = M_1 \oplus_k M_2$ (where $k \in \{2,3\}$) with $E(M_1) \cap E(M_2) = A$.  If $Z = \{z_1, z_2, z_3\}$ is a circuit in $M$ with $z_1 \in E(M_1)$ and $z_2, z_3 \in E(M_2)$ then there is an element $a \in A$ such that $\{z_1, a\}$ is a circuit of $M_1$ and $\{z_2, z_3, a\}$ is a circuit of $M_2$.
\end{lemma}
\begin{proof}
Since $Z$ is a circuit in $M$, by definition it equals $X_1 \symd X_2$ where $X_i$ is a cycle of $M_i$ for $i \in \{1,2\}$.  Thus $X_1 \setminus A = \{z_1\}$ and $X_2 \setminus A = \{z_2, z_3\}$.  We begin with the case of $k=2$.  In this case $A = \{a\}$.  If $X_1 = \{z_1\}$ then $\{z_1\}$ is a circuit of $M_1$ and thus by Lemma~\ref{lem:not_circuit} is also a circuit of $M$, which contradicts the assumption that $Z$ is a circuit of $M$.  So $X_1$ must equal $\{a,z_1\}$.  Thus $\{a, z_1\}$ is a circuit of $M_1$.  Since $Z = X_1 \symd X_2$ this implies that $X_2 = \{a, z_2, z_3\}$, so $\{a, z_2, z_3\}$ is a cycle of $M_2$.  If it is not a circuit then it can be partitioned into at least two circuits, which implies that either $\{z_1\}, \{z_2,\}$, or $\{z_1, z_2\}$ is a circuit of $M_2$ and so by Lemma~\ref{lem:not_circuit} a circuit of $M$.  This is a contradiction, since $Z$ is a circuit of $M$.  Thus $\{a, z_3, z_3\}$ is a circuit of $M_2$.

We now move to the case of $k=3$.  We first show that there is an $a \in A$ so that $\{z_1, a\}$ is a circuit in $M_1$.  Note that $X_1 \cap A \neq \emptyset$, since that would imply that $\{z_1\}$ is a circuit in $M_1$ and thus a circuit in $M$ (by Lemma~\ref{lem:not_circuit}), which would imply that $Z$ is not a circuit in $M$.  Similarly, $X_1 \cap A \neq A$ since that would obviously imply that $X_1$ is not a cycle: $X_1$ would be a strict superset of $A$ and thus not a circuit, and any partition would involve either a set with only one or two elements of $A$ (thus independent) or has $\{z_1\}$ is a circuit in $M_1$ and thus in $M$, contradicting $Z$ being a circuit.  If $|X_1 \cap A| = 2$ then $X_1$ must be a circuit, since any way of partitioning leaves an independent set consisting of either just one element of $A$ or both elements of $A \cap X_1$.  So we can apply double-elimination to $X_1$ and $A$, which implies that there is an element $a \in A$ so that there is a circuit of $M_1$ contained in $\{z_1, a\}$.  Since neither $z_1$ nor $a$ can be loops (or else $Z$ and $A$ would not be circuits), this implies that $\{z_1,a\}$ is a circuit in $M_1$ as claimed.  Finally, if $|X_1 \cap A| = 1$ then $X_1 = \{z_1, a\}$ so we are finished for the same reason.

Now we show that $\{z_2, z_3, a\}$ is a circuit of $M_2$.  We already showed that $|X_1 \cap A|$ is either $1$ or $2$, and clearly $X_2 \cap A = X_1 \cap A$.  If $|X_2 \cap A| = 1$ then $X_2 \cap A = \{a\}$, so $\{z_2, z_3, a\}$ must be a cycle.  And it is easy to see that it must be a circuit, since otherwise either $z_2, z_3$, or $a$ must be a loop, contradicting the facts that $Z$ is a circuit in $M$ and $A$ is a circuit in $M_2$.  If $|X_2 \cap A| = 2$, then by the previous paragraph $X_2$ must consist of $z_2, z_3$, and the two elements of $A \setminus \{a\}$ (call them $a'$ and $a''$).  If $X_2$ is a circuit of $M_2$ then applying double-elimination to $X_2$ and $A$ finishes the proof.  Otherwise, it must be the case that $\{z_2, a'\}$ is a circuit of $M_2$ and $\{z_3, a''\}$ is also a circuit of $M_2$.  Recall that by Lemma~\ref{lem:binary_basics} in a binary matroid the symmetric difference of any two distinct circuits is a nonempty cycle.  Applying this to $A$ and to $\{z_2, a'\}$  implies that $\{z_2, a, a''\}$ is a circuit, and then another application to $\{z_2, a, a''\}$ and $\{z_3, a''\}$ implies that $\{z_2, z_3, a\}$ is a circuit, as desired.
\end{proof}

Since an important part of our algorithm is modifying decompositions, we will need the following lemma as a basic tool.  It allows us to add an element that is parallel to an existing element.

\begin{lemma} \label{lem:adding_element}
Let $M$ be a binary matroid, let $a \in E(M)$ so that $a$ is not a loop, and let $z \not\in E(M)$.  Let $\mathcal C(M)$ denote the circuits of $M$.  Let $\mathcal C$ be a collection of subsets of $E(M) \cup \{z\}$ consisting of the union of the following three collections of sets:
\begin{enumerate}
\item $\mathcal C(M)$,
\item $\{\{z,a\}\}$, and
\item $\{(C \setminus \{a\}) \cup \{z\} : C \in \mathcal C(M), a \in C\}$.
\end{enumerate}
Then $\mathcal C$ forms the circuits of a binary matroid $M(z,a)$.
\end{lemma}
\begin{proof}
Consider the binary matroid defined by using a $GF(2)$ representation of $M$ and assigning $z$ the same vector as $a$.  Then $\mathcal C$ is exactly the circuit set of this matroid.
\end{proof}

Note that this method of adding parallel edges is essentially unique: it is not hard to see that given a matroid $M$ and an element $e \in E(M)$, adding an element $z$ and including $\{z,e\}$ as a circuit uniquely determines the independent sets and circuits  containing $z$.  Moreover, it is also easy to see that the relation of being parallel defines equivalence classes which partition the elements.

We will now use Lemma~\ref{lem:adding_element} to show that in a sum $M = M_1 \oplus M_2$, if an element of $M$ is parallel to an element of $E(M_1) \cap E(M_2)$ in $M_1$, then we can ``move" it to $M_2$ in the obvious way without affecting the decomposition.

\begin{lemma} \label{lem:moving_parallel}
Let $M = M_1 \oplus_k M_2$ for $k \in \{2,3\}$ with $E(M_1) \cap E(M_2) = A$.  Let $z \in E(M_1) \setminus A$ and $a \in A$ be elements such that $\{z,a\}$ is a circuit of $M_1$.  Then $M = (M_1 \setminus \{z\}) \oplus_k M_2(z,a)$.
\end{lemma}
\begin{proof}
Let $C$ be a circuit of $M$, and thus by definition equal to $X_1 \symd X_2$ where $X_1$ is a cycle of $M_1$ and $X_2$ is a cycle of $M_2$.  If $z \not\in C$ then obviously $C$ is also a circuit of $(M_1 \setminus \{z\}) \oplus_k M_2(z,a)$, so suppose that $z \in C$.  Then $z \in X_1 \setminus A$.  Fix any partition of $X_1$ into circuits of $M_1$, and let $B \subseteq X_1$ be the circuit of $M_1$ containing $z$ in this partition.  If $a \not\in X_1$, then let $X'_1 = X_1 \setminus \{z\} \cup \{a\}$.  Clearly $X'_1$ is a cycle of $M_1 \setminus \{z\}$: by the circuit elimination axiom $B \setminus \{z\} \cup \{a\}$ is a circuit of $M_1$ and thus of $M_1 \setminus \{z\}$, and therefore $X'_1$ is a cycle of $M_1 \setminus \{z\}$.  Let $X'_2 = X_2 \cup \{a,z\}$.  Since $X_2$ is a cycle of $M_2$ and $\{a,z\}$ is a circuit of $M_2(z,a)$, we know that $X'_2$ is a cycle of $M_2(z,a)$.  Thus $X'_1 \symd X'_2 = X_1 \symd X_2 = C$ is a cycle of $(M_1 \setminus \{z\}) \oplus_k M_2(z,a)$.

On the other hand, suppose that $a \in X_1$ (and so $a \in X_2$ as well).  By Lemma~\ref{lem:partition} we can assume that there is a partition of $X_1$ into circuits of $M_1$ so that $\{a,z\}$ is a circuit in this partition, and thus $\{a,z\} = B$.  Let $X'_1 = X_1 \setminus \{a,z\}$; clearly $X'_1$ is a cycle of $M_1 \setminus \{z\}$.  Let $X'_2 = X_2 \setminus \{a\} \cup \{z\}$.  Let $D' \subseteq X_2$ be the circuit of $M_2$ that contains $a$ in the decomposition of $X_2$ into circuits.  Then by definition $D' \setminus \{a\} \cup \{z\}$ is a circuit of $M_2(z,a)$, and thus $X'_2$ is a cycle of $M_2(z,a)$.  Thus $X'_1 \symd X'_2 = X_1 \symd X_2 = C$ is a cycle of $(M_1 \setminus \{z\}) \oplus_k M_2(z,a)$.

So now we know that any circuit of $M = M_1 \oplus_k M_2$ is also a cycle (and thus dependent) in $(M_1 \setminus \{z\}) \oplus_k M_2(z,a)$.  We now want to prove that a circuit of $(M_1 \setminus \{z\}) \oplus_k M_2(z,a)$ is dependent in $M$, which would clearly imply the lemma.  But we've essentially already done this:  it is easy to see that $M_1 = (M_1 \setminus \{z\})(z,a)$ and $M_2 = M_2(z,a) \setminus \{z\}$, so we can apply the above argument where we replace $M_1$ by $M_2(z,a)$ and replace $M_2$ by $M_1 \setminus \{z\}$.  Thus by the above argument applied to $M_2(z,a) \oplus_k (M_1 \setminus \{z\})$, any circuit of $ M_2(z,a) \oplus_k (M_1 \setminus \{z\})$ is also a circuit of $ (M_2(z,a) \setminus \{z\}) \oplus_k (M_1 \setminus \{z\})(z,a) = M_2 \oplus_k M_1 = M_1 \oplus_k M_2$.
\end{proof}

\section{Modification of Decomposition}

In this section we show how to modify the decomposition in order to ensure that certain desirable properties hold.  In particular, we will first show how to modify a given decomposition in order to make the conflict graph a forest.  We will actually prove the slightly stronger condition that for any $3$-sum (which by definition is summed along a circuit of size $3$), the circuit in the intersection is contained entirely in two of the lowest-level matroids.  In other words, while the process of summing matroids might create new circuits that contain elements that started out in different matroids, any circuit that is used as the intersection of a sum existed from the very beginning.

Once the conflict graph of the decomposition is a forest, arbitrarily rooting each tree in the forest provides a way of coordinating which matroid should contract which elements so that each matroid contracts either a single element (corresponding to a $2$-circuit) or three elements that form a circuit (corresponding to a $3$-circuit).  While this seems like a small change, in order to ensure that it does not cost us much in the value we need to modify the decomposition a little bit more in order to ensure that there are no elements in parallel with any of the contracted elements.  This can easily be accomplished by using Lemma~\ref{lem:moving_parallel} to move any parallel element to the other matroid in the sum.  We now formalize these modifications.

Suppose we are given a binary matroid $\tilde M$ and  a $\{1,2,3\}$-decomposition $(T,\mathcal M)$ of $\tilde M$.  Let $V_T = V(T)$ be the vertex set of $T$, and let $\tilde V_T = V_T \setminus \mathcal M$ denote the internal vertices of this tree.  We will sometimes abuse notation and refer to both a vertex in $V_T$ and the matroid that it represents by the same value, e.g.~$M$.

Note that each we can associate each internal vertex in $\tilde V$ not just with a matroid $M$, but also with the set of elements that is the intersection of the ground sets of its children (and is thus not in the ground set of $M$).  This set is either the empty set (if $M$ is the $1$-sum of its children), a single element (if it is the $2$-sum), or three elements that form a circuit in both of its children (if it is the $3$-sum).  For an internal node $M$, let $Z_M$ denote this set.
We will say that a set in $\mathcal Z$ is \emph{summed over} or \emph{summed along}, and call such a set a \emph{sum-set}.

Let $Z = Z_M = \{z_1, z_2, z_3\}$ be a sum-set for an internal vertex $M$ with children $M_1$ and $M_2$ (since $|Z_M| = 3$ we know that $M = M_1 \threeplus M_2$).  We associate two other matroids in $V_T$ with $Z$, one of which will be a descendant of $M_1$ and the other of which will be a descendant of $M_2$.  By definition $Z \subseteq E(M_1)$ and $Z \subseteq E(M_2)$, and since each element of a matroid is an element in exactly one of its children each element of $Z$ appears in exactly two basic matroids, one of which is a descendant of $M_1$ and the other of which is a descendant of $M_2$.  Let $M_1(Z)$ be the least-common ancestor of the basic matroids descended from $M_1$ that contain $Z$, i.e.~$M_1(Z)$ is the lowest-level descendent of $M_1$ that actually contains all of $Z$ (so one element of $Z$ is in one child of $M_1(Z)$ and the other two are in the other child).  Similarly, let $M_2(Z)$ be the LCA of the basic matroids descended from $M_2$ that contain $Z$.  We call $M_1(Z)$ and $M_2(Z)$ the \emph{creation matroids} of $Z$.

For an internal vertex/matroid $M \in \tilde V_T$, we say that a set of elements of $M$ is \emph{basic} if the elements are all from the same basic matroid; otherwise the set is \emph{non-basic}.  A sum-set $Z_M$ is in $\mathcal B$ if it has size $3$ (and thus is a circuit of size $3$ in both of the children of $M$) and if it is non-basic in at least one child of $M$.  We say that such a sum-set is \emph{bad}, and we call a $\{1,2,3\}$-decomposition \emph{good} if no sum-set is bad.  In other words, a sum-set is bad if at least one of its creation matroids is not a basic matroid, and a good decomposition is one in which the creation matroids for every sum-set are both basic matroids.

Our first lemma is that we can always modify the decomposition to become good.  The intuition behind the proof is simple: if there is a bad sum-set $Z$, then we consider a creation matroid for $Z$ that is not basic (at least one of the two creation matroids must be non-basic by the definition of a bad sum-set).  Let $M$ be this creation matroid with children $M_1$ and $M_2$, so that one element $z_1 \in Z$ is in $M_1$ and the other two elements $z_2, z_3$ are in $M_2$.  We then ``move" $z_1$ from $M_1$ to $M_2$: Lemma~\ref{lem:3circuit_creation} implies that there is some element $a \in E(M_1) \cap E(M_2)$ so that $z_1$ and $a$ are parallel in $M_1$, so we can use Lemma~\ref{lem:moving_parallel} to move $z_1$ from $M_1$ to $M_2$.  Now $Z$ is entirely present in $M_2$, which is a child of $M$, so we have made progress: if $Z$ is now basic then we have one less bad sum-set, and if not then we have moved its creation ``down" so we can repeat this until it eventually becomes basic.  We make this formal in the next lemma by formalizing this notion of ``progress".

\begin{lemma} \label{lem:no_bad}
Given a $\{1,2,3\}$-decomposition $T$, in polynomial time we can find a $\{1,2,3\}$-decomposition that has no bad sum-sets and in which all of the basic matroids are either unchanged or have some elements removed and some additional elements added parallel to existing elements.
\end{lemma}
\begin{proof}
Let $f: V_T \rightarrow [|V_T|]$ be a bijection between $V_T$ and $[|V_T|]$ (i.e.~a numbering of the vertices of $T$) with the property that $f(M)$ is larger than both $f(M_1)$ and $f(M_2)$, where $M_1$ and $M_2$ are the children of $M$ in $T$.  We can intuitively think of $f(M)$ as the ``time" that $M$ is created by taking the sum of its children.  For each bad sum-set $Z$ let $g(Z) = f(M) + f(M')$, where $M$ and $M'$ are the two creation matroids for $Z$.  We let $\Phi(T) = \sum_{Z \in \mathcal B} g(B)$ be a potential function, and note that it is always at least $0$ (by convention the sum over an empty set is $0$) and is at most $|\tilde V_T| \cdot 2|V_T|$ (since $|\mathcal B| \leq |\tilde V_T|$ and each set $Z \in \mathcal B$ has $g(Z) \leq 2|V_T|$).

We will show how to decrease $\Phi$ by at least $1$ as long as $\Phi$ is greater than zero, which obviously implies that a polynomial number of repetitions will result in $\mathcal B$ being empty.  Let $Z \in \mathcal B$ be an arbitrary bad sum-set.  Let $Z = \{z_1, z_2, z_3\}$.  Let $M$ be a creation matroid of $Z$ that is not a basic matroid (such a creation matroid must exist by the definition of $\mathcal B$), and let its children be $M_1$ and $M_2$.  By the definition of $M$ we know that $Z \not \subseteq E(M_1)$ and $Z \not\subseteq E(M_2)$, so without loss of generality we will assume that $z_1 \in E(M_1)$ and $z_2, z_3 \in E(M_2)$.  Note that since $Z$ is a sum-set there is an ancestor $M'$ of $M$ for which $Z = Z_{M'}$, and thus $Z$ is a circuit in the child of $M'$ that is an ancestor of $M$.  Then Lemma~\ref{lem:not_circuit} inductively implies that $Z$ is a circuit of $M$.  In particular, this implies that that $M$ cannot be the $1$-sum of $M_1$ and $M_2$.   Lemma~\ref{lem:3circuit_creation} then implies that there is an element $a \in Z_M$ such that $\{z_1, a\}$ is a circuit in $M_1$ and $\{z_2, z_3, a\}$ is a circuit in $M_2$.

Our modification of the decomposition involves moving $z_1$ from $M_1$ to $M_2$.  Let $M'_1 = M_1 \setminus \{z_1\}$, and let $M'_2 = M_2(z_1,a)$ (recall this involves adding $z_1$ to be parallel to $a$, as detailed in Lemma~\ref{lem:adding_element}).  Lemma~\ref{lem:moving_parallel} implies that $M'_1 \oplus M'_2 = M$, i.e.~this change did not affect the sum.

We now have to define the decompositions of $M'_1$ and $M'_2$, but this is trivial.  For $M'_1$, we simply remove $z_1$ from all of its descendants.  For $M'_2$, note that there is a unique path in $T$ from $M_2$ to the basic matroid that contains $a$, and all matroids on this path also contain $a$.  We simply include $z_1$ in these matroids in the obvious way, as we did with $M_2$: if $H$ is a matroid in this path then we set $H' = H(z_1,a)$.  It is easy to see that this forms a valid $\{1,2,3\}$-decomposition of $M'_2$, and thus we still have a valid $\{1,2,3\}$-decomposition $T'$ of $\tilde M$.  Note that $T'$ is isomorphic to $T$, i.e.~the vertices of the decomposition are unchanged; the only thing that has changed is the matroid corresponding to each vertex.  Thus we can use the same function $f$ on the vertices of $T'$.

To see that we made progress, consider $\Phi(T')$ of this new decomposition.  Note that $\mathcal Z$ is unchanged: any set that was a sum-set is still a sum-set, and vice versa.  And since all we did was move $z_1$ around, any bad sum-set other than $Z$ is still a bad sum-set, and any sum-set that was not bad is still not bad.  Moreover, the creation matroids for each bad sum-set are still the same vertices in the decomposition, so every bad sum-set other than $Z$ has the same $g$-value.  Thus if $Z$ is no longer in $\mathcal B$ then obviously $\Phi(T') < \Phi(T)$.  Otherwise, observe that in $T'$ the $g$-value of $Z$ has decreased, since of the two creation matroid for $Z$ one of them is unchanged and the other has changed from $M$ to either $M'_2$ or some descendent of $M'_2$, and $f(M'_2) = f(M_2) < f(M)$.  Thus $\Phi(T') < \Phi(T)$.
\end{proof}

A possibly simpler proof that we can get a good decomposition for regular matroids would be to analyze the algorithm for constructing such a decomposition, which works by checking whether every basic matroid is either graphic, cographic, or isomorphic to $R_{10}$, and if not then finding a way of decomposing it further.  It is easy to see that we can maintain the invariant that no sum-sets are bad.  However, this has the downside of being particular to the algorithm for constructing a decomposition of a regular matroid.  Our lemma is more general, since it allows us to work on matroids that (for example) are represented by giving us the decomposition.  This might be the case if we are working with basic matroids that form a class for which there is no efficient testing algorithm.

The following lemma gives the major reason why we want a good decomposition: it implies that the conflict graph is a forest.

\begin{lemma} \label{lem:forest}
If $T$ is a good $\{1,2,3\}$-decomposition of $\tilde M$, then $G_T$ is a forest.
\end{lemma}
\begin{proof}

To see that $G_T$ is a forest, we will prove by induction a slightly stronger property.  Let $M$ be a vertex of $T$, and let $G_M$ denote the subgraph of $G_T$ induced by the basic matroids that are descendants of $M$ in $T$.  We claim that for every vertex $M$ of $T$, the graph $G_M$ is a forest, and thus $G_T = G_{\tilde M}$ is a forest.  For the base case, this is obviously true if $M$ is leaf of $T$, since then it is a basic matroid and so $G_M$ is just the single vertex $M$.  Now let $M$ be an internal vertex with children $M_1$ and $M_2$.  By induction $G_{M_1}$ and $G_{M_2}$ are forests, so it will be enough to prove that there is a single edge between $G_{M_1}$ and $G_{M_2}$ in $G_M$.

It is easy to see that if $e$ is an element of $E(\tilde M)$ then $e$ is in exactly one basic matroid, and if $e$ is an element that is in a basic matroid but not in $E(\tilde M)$ then $e$ is an element of exactly two basic matroids, and is present in all matroids in $V_T$ that are on the two paths from these basic matroids to their least common ancestor in $T$ (other than the LCA itself).  This implies that an element $e$ is in $E(M_1) \cap E(M_2)$ if and only if it is present in a basic matroid descended from $M_1$ and in a basic matroid descended from $M_2$.  Thus if a basic matroid descended from $M_1$ and a basic matroid descended from $M_2$ share an element (i.e.~there is an edge between them in $G_M$), then this element is in $E(M_1) \cap E(M_2)$.  So if there are at least two edges between $G_{M_1}$ and $G_{M_2}$ in $G_M$ then $E(M_1) \cap E(M_2)$ contains multiple elements, which are not all contained in the same two basic matroids.  This is implies that $E(M_1) \cap E(M_2)$ is a bad sum-set, contradicting our assumption that the decomposition is good.
\end{proof}

Now that we know $G_T$ is a forest we can root each tree in it arbitrarily, and for every basic matroid $M$ that is not a root define $p(M)$ to be its parent.  For each basic matroid $M$ that is not a root, let $A_M$ be the set that corresponds to the edge between $M$ and $p(M)$.  So $A_M$ is either empty (if the edge corresponds to a $1$-sum), a single element in $M$ (if the edge corresponds to a $2$-sum), or a circuit of size $3$ in $M$ (if the edge corresponds to a $3$-sum).  Note that $A_M = E(M) \cap E(p(M))$.  The final modification to the decomposition that we will need is to make sure that there are no $2$-circuits in $M$ that involve an element of $A_M$.

\begin{theorem} \label{thm:structure_main}
Given a binary matroid $\tilde M$ and a $\{1,2,3\}$-decomposition $\tilde T$ with basic matroids $\tilde {\mathcal M}$, we can in polynomial time find a good $\{1,2,3\}$-decomposition $T$ with basic matroids $\mathcal M$ such that for every basic matroid $M \in \mathcal M$ there are no circuits of size $2$ that include an element of $A_M$ and an element of $E(\tilde M)$.  Furthermore, every basic matroid $M \in \mathcal M$ can be obtained by taking a $M' \in \tilde {\mathcal M}$, deleting elements, and adding parallel elements.
\end{theorem}
\begin{proof}
Lemmas~\ref{lem:no_bad} and \ref{lem:forest} imply that we can find a $\{1,2,3\}$-decomposition $T$ so that $G_T$ is a forest and every basic matroid in $\mathcal M$ can be obtained by taking a matroid $M' \in \tilde{\mathcal M}$ and removing elements and adding parallel elements.  Let $z \in E(\tilde M)$ be an element from the sum, and let $M \in \mathcal M$ be the basic matroid tat contains it.  If there is an element $a \in A_M$ such that $\{z,a\}$ is a circuit of $M$, then we say that $z$ is \emph{bad}.  We modify the decomposition by moving $z$ to $p(M)$: we redefine $M$ to be $M \setminus \{z\}$ and redefine $p(M)$ to be $p(M)(z,a)$ (as in Lemma~\ref{lem:adding_element}).  This changes the internal vertices of $T$ in the obvious way: any ancestor $Q$ of $M$ in $T$ that is not an ancestor of $p(M)$ becomes $Q \setminus \{z\}$ and any ancestor $Q$ of $p(M)$ that is not an ancestor of $M$ becomes $Q(z,a)$.  By repeated applications of Lemma~\ref{lem:moving_parallel} this still gives us a valid $\{1,2,3\}$-decomposition of $\tilde M$, and since only $z$ (an element of $E(\tilde M)$, and thus in only one basic matroid) is being moved, $G_T$ is still the same forest that it was originally.  If $z$ is still bad (i.e.~in a circuit of size $2$ with an element of $A_{p(M)}$), then this operation can be repeated to move $z$ to $p(p(M))$, then $p(p(p(M)))$, etc.  Eventually, $z$ will either become an element of a matroid $M$ with $A_M \neq \emptyset$ in which it not a part of any circuit of size $2$ with an element of $A_M$, or it will become part of a matroid $M$ that is a root in $G_T$ (and thus has $A_M = \emptyset$).  In either case, $z$ is no longer bad.  The number of times we have to do this for $z$ is at most the depth of the tree of $G_T$ containing the basic matroid that had $z$ as an element, and thus polynomial.  Note that this process does not make any element bad that was not bad originally, so we can apply it to every bad element and in polynomial time get a $\{1,2,3\}$-decomposition with no bad elements.  Moreover, the basic matroids are only changed by moving elements from one matroid (deletion) to another (adding parallel elements), so they can still be constructed from the original basic matroid in $\tilde{\mathcal M}$ by deletions and additions of parallel elements.
\end{proof}

\section{Algorithm}

It is now simple to define our algorithm.  Given a matroid $\tilde M$ (which without loss of generality has no loops), we first apply Theorem~\ref{thm:structure_main} to get a $\{1,2,3\}$-decomposition $T$ with basic matroids $\mathcal M$.  Then for every basic matroid $M \in \mathcal M$ in the decomposition we run in parallel a secretary algorithm on $(M / A_M) |(E(M) \cap E(\tilde M))$.  In other words, for every basic matroid $M$ we contract $A_M$, ignore all of the remaining elements that were added by the decomposition and so don't appear in $\tilde M$, and run our matroid secretary algorithm.  Since every element of $E(\tilde M)$ is in exactly one basic matroid $M$, we can just run these algorithms in parallel: when we see an element, we determine to which basic matroid it belongs and use that algorithm to make our decision.  If every secretary algorithm that we use on a basic matroid is $\alpha$-competitive, then we claim that overall the algorithm is $O(\alpha)$-competitive.

We will first prove that it gives a valid independent set.  We show something slightly stronger: if we choose an independent set from $(M /A_M) |(E(M) \cap E(\tilde M))$ for every $M \in \mathcal M$, then their union is independent in $\tilde M$.  In order to prove this, we first prove a few more lemmas about matroids, matroid sums, and how they relate to contraction and deletion.

\begin{lemma} \label{lem:cycle_contract}
Let $M$ be a binary matroid and $C \subseteq E(M)$ be a cycle of $M$, and let $X \subseteq E(M)$ be an arbitrary subset of elements with $C \not\subseteq X$.  Then $C \setminus X$ is a cycle of $M / X$.
\end{lemma}
\begin{proof}
It is obvious that $M / X = (M / \{x\}) / (X\setminus\{x\})$ for any element $x \in X$ (see e.g.~\cite[Proposition 3.1.26]{Oxley}), so if we prove that $C \setminus \{x\}$ is a cycle of $M / \{x\}$ then by induction we have proved the lemma.  Similarly we may assume that $C$ is a circuit of $M$, since if it is a cycle but not a circuit then we can decompose $C$ into circuits, each of which (when $x$ is removed) becomes a cycle in $M / \{x\}$, and thus $C \setminus \{x\}$ is a cycle of $M / \{x\}$.

So let $C$ be a circuit of $M$.  If $x \not\in C$, then Corollary 9.3.7~of~\cite{Oxley} states that $C$ is either a circuit of $M/\{x\}$ or is the disjoint union of two circuits of $M/\{x\}$, and thus $C = C\setminus \{x\}$ is a cycle of $M/\{x\}$.  If $x \in C$, then $C \setminus \{x\}$ is a circuit (and thus a cycle) of $M/\{x\}$.  This is because the circuits of $M/\{x\}$ are exactly the non-empty minimal subsets of $\{A \setminus \{x\} : A\text{ a circuit of } M\}$ (see~\cite[Proposition 3.1.11]{Oxley}).
\end{proof}

The next lemma shows that contractions can be moved outside of matroid sums.

\begin{lemma} \label{lem:sum_contract1}
Let $M_1, M_2$ be binary matroids with $M = M_1 \oplus M_2$ and $E(M_1) \cap E(M_2) = Z$.  Let $X \subseteq E(M_2) \setminus Z$.  Then any independent set of $M_1 \oplus (M_2 / X)$ is also independent in $(M_1 \oplus M_2) / X$.
\end{lemma}
\begin{proof}
We prove the contrapositive.  Let $Y$ be a circuit of $(M_1 \oplus M_2) / X$.  Then for any basis $B_X$ of $X$ we know that $Y \cup B_X$ is dependent in $M = M_1 \oplus M_2$.  Thus there are subsets $Y' \subseteq Y$ (which must be nonempty) and $B' \subseteq B_X$ so that $Y' \cup B'$ is a circuit of $M_1 \oplus M_2$.  So there are cycles $Y_1$ of $M_1$ and $Y_2$ of $M_2$ so that $Y' \cup B' = Y_1 \symd Y_2$. Lemma~\ref{lem:cycle_contract} implies that $Y_2 \setminus X = Y_2 \setminus B'$ is a cycle of $M_2 / X$, and since $Z \cap X = \emptyset$ this means that $Y_2 \setminus B' \cap Z = Y_2 \cap Z = Y_1 \cap Z$, and thus $Y_1 \symd (Y_2 \setminus B')$ is a cycle of $M_1 \oplus (M_2 / X)$.  Since $B' \subseteq X \subseteq E(M_2) \setminus Z$, we know that $B' \cap Y_1 = \emptyset$ and thus $Y_1 \symd (Y_2 \setminus B') = (Y_1 \symd Y_2) \setminus B' = Y' \subseteq Y$.  So $Y$ is dependent in $M_1 \oplus (M_2 / X)$.
%
\end{proof}
%

The next lemma, despite being a rather simple observation, forms the core of our argument.  It says that if $M = M_1 \oplus M_2$, then as long as we contract $E(M_1) \cap E(M_2)$ in \emph{one} of $M_1, M_2$ we can treat them separately without worrying about independence.  This is what will let us run a different algorithm in each of the basic matroids without worrying about strange interactions.

\begin{lemma} \label{lem:sum_contract2}
Let $M_1, M_2$ be binary matroids with $M = M_1 \oplus M_2$ and $E(M_1) \cap E(M_2) = Z$.  Let $X_1$ be independent in $M_1 \setminus Z$, and let $X_2$ be independent in $M_2 / Z$.  Then $X = X_1 \cup X_2$ is independent in $M$.
\end{lemma}
\begin{proof}
Suppose that $X$ is dependent in $M_1 \oplus M_2$.  Then there is a circuit $Y \subseteq X$ that can be expressed as $Y_1 \symd Y_2$ where $Y_1$ is a cycle in $M_1$ and $Y_2$ is a cycle in $M_2$.  We break into cases depending on $|Y_2 \cap Z|$.  If $|Y_2 \cap Z| = 0$ then $Y_2 \subseteq X_2$, which is a contradiction since $X_2$ is independent in $M_2 / Z$ (and thus in $M_2$) while $Y_2$ is a cycle of $M_2$.  If $|Y_2 \cap Z| = 1$ or $|Y_2 \cap Z| = 2$, then let $Y_2 \cap Z = P$.  Then $Y_2 \setminus P \subseteq X_2$, so is independent in $M_2 / Z$.  Note that $P$ is independent in $M_2 |Z$, since if we are in the $2$-sum case $Z =P$ is just a single non-loop element, while if we are in the $3$-sum case $Z$ is a circuit of size $3$ and $P \subset Z$.  Thus by the definition of contraction $(Y_2 \setminus P) \cup P = Y_2$ is independent in $M_2$, giving a contradiction.

Finally, suppose that $|Y_2 \cap Z| = 3$.  Since $Y_2$ is a cycle in $M_2$, it can be partitioned into disjoint circuits.  If $Z$ is not a circuit in the partition then there is a circuit $C$ in the partition with at most $2$ elements from $Z$.  This is a contradiction, since $C \setminus Z \subseteq Y_2 \setminus Z \subseteq X_2$ is independent in $M_2 / Z$, so adding any $2$ elements of $Z$ to $X_2$ gives an independent set in $M_2$ (since $Z$ is a circuit of size $3$).  Thus $Z$ is a circuit in the partition, but since $Y_2 \setminus Z \subseteq X_2$ is independent in $M_2 / Z$ it is also independent in $M_2$, so there cannot be any other circuits in the partition.  So $Y_2 = Z$.  But this means that $Y = Y_1 \symd Y_2 = Y_1 \setminus Z$.  Since $Y_2 \cap Z = Y_1 \cap Z$, this in fact implies that $Y_1 = Y \cup Z$.

We claim that this is a contradiction.  Note that $Y \subseteq X_1$ is independent in $M_1 \setminus Z$ and thus in $M_1$.  Since $Y_1$ is a cycle of $M_1$, it can be partitioned into disjoint circuits of $M_1$.  Then by Lemma~\ref{lem:binary_basics} the symmetric difference of $Y_1$ and $Z$ is either empty or contains a circuit.  But $Y_1 \symd Z = Y$ is nonempty and independent in $M_1$, giving a contradiction.
\end{proof}

We now prove that by contracting $A_M$ in each basic matroid $M$ and ignoring all other elements of the basic matroid that are are not in $E(\tilde M)$ (i.e.~fake elements that were introduced by the decomposition), we can choose any independent we want in each basic matroid and still be guaranteed global independence.  In other words, the sets $A_M$ give us the global coordination that we require.

\begin{lemma} \label{lem:independence}
For each $M \in \mathcal M$, let $I_M$ be an independent set of $(M / A_M) |(E(M) \cap E(\tilde M))$.  Then $I = \cup_{M \in \mathcal M} I_M$ is independent in $\tilde M$.
\end{lemma}
\begin{proof}
We prove this by induction on $T$.  Recall that every vertex in $T$ corresponds to a matroid consisting of the sum of its children and to at most one edge of $G_T$ (with it not corresponding to any edge in $G_T$ if and only if it represents a matroid that is the $1$-sum of its children).  For an internal vertex $H$ of $T$, recall that $Z_H$ is the intersection of the ground sets of its children.  Also, for every vertex $H$ of $T$ let  $D(H)$ be the set of basic matroids that are its descendants, and let $B_H = \cup_{M \in D(H) : A_M \subseteq E(H)} A_M$.  Less formally, $B_H$ is the set of elements that some basic matroid that is a descendent of $H$ contracts when running the algorithm, except for those elements that were already summed along by a descendent of $H$ and thus are not elements of $H$.
We claim that  $\cup_{M \in D(H)} I_M$ is independent in $(H / B_H)$ for every vertex $H \in T$.  This gives the lemma when applied to the root $H = \tilde M$ since $B_{\tilde M} = \emptyset$.

We prove this by induction.  For the base case, let $H$ be a basic matroid.  Then $\cup_{M \in D(H)} I_M = I_H$, which by definition is independent in $(H / A_H) = (H / B_H)$.  For the inductive case, let $H$ be the parent of $H_1$ and $H_2$.  Note that since there are no bad sets in the decomposition $Z_H$ is contained in exactly one $M_1 \in D(H_1)$ and in exactly one $M_2 \in D(H_2)$, and is equal to either $A_{M_1}$ or $A_{M_2}$ (depending on whether $M_1 = p(M_2)$ or $M_2 = p(M_1)$ in $G_T$).  Without loss of generality we will assume that $Z_H = A_{M_1}$, and thus $Z_H \subseteq B_{H_1}$ and $Z_H \cap B_{H_2} = \emptyset$.  Thus $B_H = (B_{H_1} \setminus Z_H) \cup B_{H_2}$.  By induction $\cup_{M \in D(H_1)} I_M$ is independent in $H_1 / B_{H_1} = (H_1 / (B_{H_1} \setminus Z_H)) / Z_H$.  Similarly, by induction and the fact that every $I_M$ contains only elements from $E(\tilde M)$ we have that $\cup_{M \in D(H_2)} I_M$ is independent in $H_2 / B_{H_2} =(H_2 / B_{H_2}) \setminus Z_H$.  Thus by Lemma~\ref{lem:sum_contract2} we have that $\cup_{M \in D(H)} I_M$ is independent in $(H_1 / (B_{H_1} \setminus Z_H)) \oplus (H_2 / B_{H_2})$.  Now we can apply Lemma~\ref{lem:sum_contract1} to move the contracted sets outside of the sum, so after two applications of the lemma we get that $\cup_{M \in D(H)} I_M$ is independent in $(H_1 \oplus H_2) / ((B_{H_1} \setminus Z_H) \cup B_{H_2}) = (H_1 \oplus H_2) / B_H$ as claimed.
\end{proof}

Now it just remains the show that our algorithm returns a set with large value.  We first will prove a lemma on the cost of contractions. Recall that for any matroid $M$ with weights on the elements, we let $OPT(M)$ denote the independent set of $M$ with the most total weight.  Since our weight function $w$ is only on elements of $E(\tilde M)$, in a basic matroid $M \in \mathcal M$ elements that are not in $E(\tilde M)$ do not have a weight.  We slightly abuse notation by extending $w$ to give these elements weight $0$.

Intuitively, since we only contract at most three elements of each basic matroid (and in the case of contracting three elements they must be a circuit) the maximum weight independent set after contraction should have close to the same weight as the maximum independent set before contraction.  At a first glance, this only fails when there is an element that is parallel to a contracted element, which is exactly why we modified the decomposition to ensure that this cannot be the case.  But it is worth noting that, like most other parts of our algorithm, this only works for binary matroids.  In the uniform matroid $U_{2,4}$ (the matroid on four elements in which any two form a basis), if we contract a circuit of size $3$ the remaining element is not independent, while before contraction it is.  Thus the adversary could place all of its value on this one element, and it would be independent but we would not be able to include it.  On the other hand, it is easy to see that this type of situation cannot happen in a graphical matroid.  This example is instructive, since $U_{2,4}$ is the forbidden minor for a matroid being binary, and the proof of the following lemma can be seen as an extension of the natural proof of this property for graphs.

\begin{lemma} \label{lem:value_contract}
Let $M \in \mathcal M$ be a basic matroid.  Then $w(OPT(M / A_M)) \geq (1/3) \cdot w(OPT(M))$.
\end{lemma}
\begin{proof}
If $A_M = \emptyset$ then the lemma is trivial, and if $OPT(M / A_M) = OPT(M)$ then it is also trivial.  So we will assume that $|A_M|$ is either $1$ (for a $2$-sum) or $3$ (for a $3$-sum) and that $OPT(M)$ is dependent in $M / A_M$ (note that $OPT(M)$ cannot contain any element of $A_M$ as $OPT(M) \subseteq E(\tilde M)$).  We know from Theorem~\ref{thm:structure_main} that there are no circuits of size $2$ in $M$ that use an element of $A_M$.  If $|A_M| = 1$, then let $\{a\} = A_M$.  Since $OPT(M)$ is dependent in $M /A_M$ there is a unique circuit of $M$ in $OPT(M) \cup \{a\}$ (see e.g.~\cite[Proposition 1.1.6]{Oxley} for a proof of uniqueness), and as noted this circuit has size at least three.  So if we remove whatever element of this circuit has the smallest weight (other than $a$, which has weight $0$), we are left with a set that is independent in $M / A_M$ and has has at least half of the weight of $OPT(M)$.

If $|A_M| = 3$, let $B \subseteq A_M$ be an arbitrary set of size $2$ (so a base of $A_M$).  Let $\{x\} = A_M \setminus B$.  For $i \in [|OPT(M)|]$, let $A_i \subseteq OPT(M)$ denote the $i$ elements of $OPT(M)$ with the most weight.  We first prove that $A_1 = \{a\}$ is independent in $M / A_M$, i.e.~that $A_1 \cup B$ is independent in $M$.  Since $A_1 \subseteq OPT(M) \subseteq E(\tilde M)$, we know that $a$ does not form a circuit of size $2$ with any element of $B$, so for $A_1 \cup B$ to be dependent in $M$ it must be a circuit.  Applying double-elimination to $B \cup A_1$ and $A_M$ implies that there is a circuit contained in $\{a,x\}$.  This is a contradiction, since $M$ has no loops and no elements of $E(\tilde M)$ (in particular, $a$) that form circuits of size $2$ with any element of $A_M$ (in particular, $x$).  Thus $A_1$ is independent in $M / A_M$.

If $|OPT(M)| \leq 3$ then clearly $w(OPT(M / A_M)) \geq w(A_1) \geq (1/3) \cdot w(OPT(M))$ so we are done.  If $|OPT(M)| \geq 4$, then there is a set of $|OPT(M)| - 3$ elements of $OPT(M) \setminus A_1$ that we can add to $A_1 \cup B$ while still maintaining independence in $M$ (by the independence augmentation axiom).  This gives us a set $D$ that has $a$ (the most valuable element of $OPT(M)$) and $B$ as well as all but two elements of $OPT(M) \setminus A_1$ that is independent in $M$.  So $D \setminus B$ is independent in $M/A_M$ and $w(OPT(M/A_M)) \geq w(D \setminus B) \geq (1/3) \cdot w(OPT(M))$.
\end{proof}

Now we can finally prove our main theorem, that our algorithm is $O(\alpha)$-competitive when the algorithm that we run on each basic matroid is $\alpha$-competitive.  For a matroid $M$, let $D(M)$ denote the set of matroids that can be obtained from $M$ by removing elements, adding parallel elements, and contracting either nothing,  a single element, or a circuit of size $3$.

\begin{theorem} \label{thm:main}
Let $\tilde M$ be a binary matroid for which we can find a $\{1,2,3\}$-decomposition with basic matroids $\tilde {\mathcal M}$, and suppose that for every matroid $M \in \tilde{\mathcal M}$ and $M' \in D(M)$ there is an algorithm $\mathcal Alg_{M'}$ that is $\alpha$-competitive.  Then $\tilde M$ admits an $O(\alpha)$-competitive algorithm for the matroid secretary problem.
\end{theorem}
\begin{proof}
Recall that our algorithm first applied Theorem~\ref{thm:structure_main} to get a new $\{1,2,3\}$-decomposition $T$ with basic matroids $\mathcal M$ that can be obtained from $\tilde{\mathcal M}$ by removing elements and adding parallel elements.  Our algorithm then runs in parallel for each $M \in \mathcal M$ the algorithm $Alg_{(M / A_M) |(E(\tilde M) \cap E(M))}$.  Note that $(M / A_M) |(E(\tilde M) \cap E(M))$ is clearly in $D(M')$ for some original basic matroid $M' \in \tilde{\mathcal M}$.

On basic matroid $M \in \mathcal M$ we run a secretary algorithm that is $\alpha$-competitive, which by definition always returns an independent set in $(M / A_M)|(E(M) \cap E(\tilde M))$.  So by Lemma~\ref{lem:independence} our algorithm returns an independent set in $\tilde M$.  To see that it is within $O(\alpha)$ of the optimal solution, note that $w(OPT(\tilde M)) = \sum_{M \in \mathcal M} w(OPT(\tilde M)|(E(M))) \leq \sum_{M \in \mathcal M} w(OPT(M)) \leq \sum_{M \in \mathcal M} 3 \cdot w(OPT(M / A_M)) = \sum_{M \in \mathcal M} 3 \cdot w(OPT((M / A_M)|(E(M) \cap E(\tilde M))))$, where the final inequality is from Lemma~\ref{lem:value_contract} and the final equality is due to assigning weight $0$ to elements in $E(M) \setminus E(\tilde M)$.  Now by definition, for each basic matroid $M \in \mathcal M$ the algorithm that we run returns a set with expected weight at least $(1/\alpha) \cdot w(OPT((M / A_M)|(E(M) \cap E(\tilde M))))$.  Thus the expected weight of the entire set returned by our algorithm is at least $(1/\alpha) \sum_{M \in \mathcal M} w(OPT((M / A_M)|(E(M) \cap E(\tilde M)))) \geq \frac{1}{3\alpha} \cdot w(OPT(\tilde M))$.
\end{proof}

This gives an algorithm for regular matroids almost immediately.

\begin{corollary}
There is a $9e$-competitive algorithm for the matroid secretary problem on regular matroids.
\end{corollary}
\begin{proof}
Theorem~\ref{thm:regular_decomposition} tells us that we can find a $\{1,2,3\}$ decomposition of any regular matroid $\tilde M$ in which all basic matroid are graphic, cographic, or isomorphic to $R_{10}$.  Theorem~\ref{thm:main} then implies that we simply need to find a good algorithm for any matroid that is in $D(M)$ where $M$ is either a graphic matroid, a cographic matroid, or $R_{10}$.  Note that the class of graphic matroids is closed under deletions, additions of parallel elements, and contractions, as is the class of cographic matroids.  Moreover, there is already a $2e$-competitive algorithm for graphic matroids~\cite{KP09} and $3e$-competitive algorithm for cographic matroids~\cite{Soto11}, so we simply need to find an algorithm for matroids in $D(R_{10})$.

Let $M \in D(R_{10})$.  Note that there are no $3$-circuits in $R_{10}$, so $M$ was created by either contracting a single element or contracting nothing.  If any element of $R_{10}$ was contracted to construct $M$, then it is easy to see that we could contract this element first and then delete elements and add parallel elements to get $M$.  Furthermore, it is known that $R_{10} / e$ is isomorphic to the cographic matroid on $K_{3,3}$ for any element $e \in E(R_{10})$ (see~\cite[Lemma 11.2.8]{Oxley}).  Thus $M$ is cographic, so there is a $3e$-competitive algorithm for it.  On the other hand, assume that no element is contracted when constructing $M$.  If there are less than $10$ parallel element classes, then $M$ can be constructed by first deleting elements of $R_{10}$ and then adding parallel elements.  But for any element $e \in E(R_{10})$ it is known that $R_{10} \setminus \{e\}$ is isomorphic to the graphic matroid on $K_{3,3}$ (see~\cite[Lemma 11.2.8]{Oxley}), and thus $M$ is graphic so there is a $2e$-competitive algorithm.

So assume that $M$ has $10$ parallel element classes.  We choose one of them uniformly at random and delete all elements in it.  Then the resulting matroid is graphic, so there is a $2e$-competitive algorithm for it.  Since the optimal solution only includes one element of each parallel class, in expectation removing one parallel class at random only loses $1/10$ of the value of $OPT$.  Thus this algorithm is $2e \cdot (10/9) = 20e / 9$-competitive.

Thus applying Theorem~\ref{thm:main} to $\tilde M$ and the decomposition we get from Seymour's decomposition theorem gives us a competitive ratio of $3 \cdot \max\{2e, 3e, 20e/9\} = 9e$.
\end{proof}

We can go beyond regular matroid to MFMC matroids, since Seymour showed that any MFMC matroid can be constructed by taking $1$- and $2$-sums of regular matroids and the Fano matroid $F_7$, and thus any MFMC matroid has a $\{1,2,3\}$-decomposition in which all of the basic matroids are either graphic, cographic, or isomorphic to either $R_{10}$ or $F_7$.

\begin{theorem}
There is a $9e$-competitive algorithm for the matroid secretary problem on MFMC matroids.
\end{theorem}
\begin{proof}
Clearly by Theorem~\ref{thm:main} and the above discussion we just need to design a $3e$-competitive algorithm for matroids in $D(F_7)$.  Our approach is similar to the one we took for $R_{10}$.  Let $M \in D(R_{10})$.  If $M$ has less than seven parallel classes then $M$ can be  constructed by either deleting or contracting an element of $F_7$ before adding parallel elements (and possibly deleting other elements).  But any single-element deletion or contraction of $F_7$ is isomorphic to the graphic matroid of $K_4$ (see~\cite[Section 1.5]{Oxley}), so in this case we already have a $2e$-competitive algorithm.  If $M$ has seven parallel classes then we choose one of them uniformly at random to delete, obtaining a graphic matroid in which the optimal solution has at least $6/7$ of the value of the original optimal solution (in expectation).  This gives a $(7/6) \cdot 2e < 3e$-competitive algorithm, so we get a $3e$-competitive algorithm for MFMC matroids.
\end{proof}

\section{Future Work}
While our motivation was regular matroids, our techniques imply that any matroid with a $\{1,2,3\}$-decomposition in which the basic matroids admit good algorithms for matroid secretary, itself admits a good algorithm.  We used this to extend our results to MFMC matroids, but it would be interesting to see how much further this can be pushed -- what other classes of matroids admit such decompositions?  There is also a significant amount of work (see e.g.~\cite{Tru92}) on extending the definition of matroid sums to non-binary matroids (Seymour's original definition, which we use, breaks down for non-binary matroids), so it would be interesting to extend our techniques to these more general definitions of sums.

\bibliographystyle{alpha}
\bibliography{matroids}

\end{document}